\documentclass{ijuc}
\usepackage{microtype}
\usepackage{multirow}
\usepackage{amssymb}
\usepackage{amsmath}
\usepackage{mathtools}
\usepackage[utf8]{inputenc}
\usepackage{graphicx}
\usepackage{amsfonts}
\usepackage{color}
\usepackage{bm}
\usepackage{xspace,soul}
\usepackage{tcolorbox}
\usepackage{calc}
\usepackage{ifthen}

\usepackage{cleveref} 
\crefname{equation}{}{}

\usepackage{amsthm}

\protected\def\verythinspace{%
  \ifmmode
    \mskip0.3\thinmuskip
  \else
    \ifhmode
      \kern0.050004em
    \fi
  \fi
}

\makeatletter
\def\sdef#1{\expandafter\def\csname#1\endcsname}
\newcommand{\checkifloaded}[1]{\@ifpackageloaded{#1}{\sdef{DidWeLoaD#1}{yes, #1 is loaded}}{\sdef{DidWeLoaD#1}{no, #1 is not loaded}}}
\makeatother
\checkifloaded{mathtools}

\newtheorem{theorem}{Theorem}[section]
\newtheorem{corollary}{Corollary}[theorem]
\newtheorem{lemma}[theorem]{Lemma}
\theoremstyle{definition}
\newtheorem{definition}{Definition}[section]

\begin{document}

\title{Finding ground states of a square Newman-Moore lattice with sides equal to a Mersenne number using the Rule 60 cellular automaton}

\author{Jon\'{a}s Carmona-P\'{i}rez\inst{1}
\and Adrian J. Peguero \inst{2}
\and Vanja Dunjko\inst{2}
\and Maxim Olshanii\inst{2} 
\and Joanna Ruhl\inst{2}\email{joanna.ruhl001@umb.edu}
}

\institute{Instituto de Biomedicina de Sevilla, Hospital Universitario Virgen del Roc\'{i}o/Consejo Superior de Investigaciones Cient\'{i}ficas/Universidad de Sevilla, 41013 Seville, Spain
\and
Department of Physics, University of Massachusetts Boston, Boston Massachusetts 02125, USA
}

\def\received{Received June 17, 2025}

\maketitle

\begin{abstract}
We offer detailed proofs of some properties of the Rule 60 cellular automaton on a ring with a Mersenne number circumference. We then use these properties to define a propagator, and demonstrate its use to construct all the ground state configurations of the classical Newman-Moore model on a square lattice of the same size. In this particular case, the number of ground states is equal to half of the available spin configurations in any given row of the lattice.  
\end{abstract}


\section{Introduction}

Both the classical Newman-Moore model (a particular two-dimensional spin lattice model with a specific three-spin interaction) \cite{newman1999_5068} and its quantum counterpart \cite{yoshida2014quantumcriticalityisingmodel} play an important role in studies of systems with glassy dynamics (see, for example,\cite{PhysRevE.72.016103, PhysRevE.62.3404} and references therein). Ground states of this 2D model are related to the time propagation of a 1D Rule 60 cellular automaton \cite{sfairopoulos2023_174107}. In this paper, we  will first recap some standard definitions used throughout the paper, then offer detailed proofs of some number-theoretic properties of the Rule 60 cellular automaton on a ring with Mersenne number circumference. Finally, we will use this knowledge to define a propagator which can be used to infer all ground state configurations of the classical Newman-Moore model on a square lattice of the same size.

\section{Rule 60 cellular automaton}
For the binary one-dimensional cellular automaton \cite{martin1984_219}, 
the state of the system at the instant of time $j$ is given by a length $L$ string: 

\begin{align}
\vec{x}_{j} = [x_{0,j},\,x_{1,j},\,\ldots\,\,x_{i,j},\,\ldots,\,x_{L-1,j}]\,\,,
\label{state_structule}
\end{align}
where $x_{i,j}$ is the state of the $i$'th bit at time $j$. Each $x_{i,j}$ is either $0$, if the site is unoccupied, or $1$, if the site is occupied. 

Consider a Rule 60 automaton \cite{martin1984_219} on a ring of a circumference $L$. According to the rule, with this periodic boundary condition the state of the automaton at the instance of time $j+1$ is related to its state at time $j$ by 
\begin{align}
\begin{array}{lcc}
x_{i,j+1} = \text{Xor}(x_{i-1,j},\,x_{i,j}) & \text{ for } & i=1,\,2,\,\ldots,\,L-1
\\
x_{0,j+1}=\text{Xor}(x_{L-1,j},\,x_{0,j}) & \text{ for } & i=0 \qquad \qquad \qquad
\\
j= 0,\,1,\,\ldots \quad.
\end{array}
\label{rule_60}
\end{align}
\subsection{Periodic trajectories}
In what follows, we will be interested in periodic trajectories generated by the rule \eqref{rule_60}.  
\begin{definition}
An initial condition $\vec{x}_{0}$ initiates a \textbf{$\bm{M}$-periodic trajectory} $\left\{\vec{x}_{j}\right\}_{j=0}^{\infty}$
if 
\begin{align*}
&
\vec{x}_{M} = \vec{x}_{0}
\,\,.
\end{align*}
\label{def:period}
\end{definition}
\begin{definition}
An initial condition $\vec{x}_{0}$ initiates a \textbf{$\bm{M}$-fundamental-periodic trajectory} $\left\{\vec{x}_{j}\right\}_{j=0}^{M-1}$
if 
\begin{align*}
&
\vec{x}_{M} = \vec{x}_{0}
\\
&
\text{and}
\\
&
\text{for any } 0<j<M, \quad \vec{x}_{j}\neq \vec{x}_{0}
\,\,.
\end{align*}
\end{definition}
\subsection{Configurations and their characterization}
The state of the system at a given instant of time is called 
a \textbf{configuration}. E.g.\ for the trajectory
\begin{align*}
\begin{array}{lll}
j=0 &
&
[100]
\\
j=1 &
&
[110]
\\
j=2 &
&
[101]
\\
j=3 &
&
[011]
\\
&\ldots&
\end{array}
\,\,,
\end{align*}
$[101]$ is the state of the size $L = 3$ lattice at the instance of time $j=2$, given that the initial condition at $j=0$ was $[100]$.
\begin{definition}
Let $_{L}N_{\text{occupied}}$-type state indicate a configuration that features $N_{\text{occupied}}$ total occupied sites on a length $L$ lattice. 
\end{definition}
E.g. $[110]$ will be characterised as a configuration of a $_{3}2$-type.

\subsection{Useful lemmas}

In this section we present some lemmas which will be used to prove later statements about the time evolution of the Rule 60 automaton.
\begin{lemma}
For any $L$, there are as many $_{L}(\text{\rm even})$-type as there are $_{L}(\text{\rm odd})$-type states.
\label{lem:total_even=total=odd}
\end{lemma}
\begin{proof}
One can establish a one-to-one correspondence between $_{L}(\text{\rm even})$ and $_{L}(\text{\rm odd})$ states simply by controlling the state of the first bit, leaving other bits unchanged. 
\end{proof}
\begin{lemma}
For any $L$, the state $[00\ldots 0]$ is a stationary state (a $1$-periodic trajectory). 
\label{lem:all_0s}
\end{lemma}
\begin{proof}
This directly follows from the rule \eqref{rule_60}.
\end{proof}
\begin{lemma}
For any ring of circumference $L$, both $_{L}(\text{\rm odd})$-type and 
$_{L}(\text{\rm even})$-type states become an $_{L}(\text{\rm even})$-type state after one time step: 
\begin{align*}
&
_{L}(\text{\rm odd}) \to _{L}(\text{\rm even})
\\
&
_{L}(\text{\rm even}) \to _{L}(\text{\rm even})
\,\,.
\end{align*}
\label{lem:odd_and_even_to_even}
\end{lemma}
\begin{proof} One proof is presented in \cite{martin1984_219} as Lemma 3.1. We present an alternative proof here.
Assume that the state of the system at an instance $j$ is 
\[
[x_{0,j},\,x_{1,j},\,\ldots\,\,x_{i,j},\,\ldots,\,x_{L-1,j}]
\,\,.\]
According to the rule \eqref{rule_60}, at the next step, $1$'s will be marking points where 
the site state changes from $0$ to $1$ or vice versa, along the lattice. For example, in 
\begin{align*}
\begin{array}{lll}
\vec{x}_{j} &=&    [011\bm{01}1001]
\\
\vec{x}_{j+1} &=& [1101\bm{1}0101]
\end{array}
\,\,,
\end{align*}
the boldface ``$1$'' at $j+1$ marks a lateral switch from ``$0$'' to ``$1$'', in bold, at $j$. Likewise in 
\begin{align*}
\begin{array}{lll}
\vec{x}_{j} &=&    [\bm{0}1101100\bm{1}]
\\
\vec{x}_{j+1} &=& [\bm{1}10110101]
\end{array}
\,\,,
\end{align*}
the boldface ``$1$'' at $j+1$ marks a lateral switch from ``$1$'' to ``$0$'' (in boldface), as the ring closes. 

Observe now that on a ring, the site state can only change laterally an \emph{even} number of times. This proves the Lemma.
\end{proof}
\begin{corollary}
Type $_{L}(\text{\rm odd})$ can not initiate a periodic motion.
\label{cor:no_periods_for_odd}
\end{corollary}
\begin{proof}
For an $M$-periodic motion, the state must return to itself after $M$ steps. But according to the 
Lemma \ref{lem:odd_and_even_to_even}, an $_{L}(\text{\rm odd})$-type state can never become 
an $_{L}(\text{\rm odd})$-type state.  
\end{proof}
\begin{lemma}
Let $M_{\text{\rm max}}$ be the maximum fundamental period for a given $L$. Let $M_{\mathcal{P}}$ be the period emerging from trajectories generated by the $_L 1$ initial state. Then 
\[M_{\text{\rm max}} = M_{\mathcal{P}} \,\,.\] Other periods divide $M_{\text{\rm max}}$.
\label{lem:maximal_M_and_shorter_periods}
\end{lemma}
\begin{proof}
The Lemma \ref{lem:maximal_M_and_shorter_periods} appears in \cite{ehrlich1990_302} as Theorem 1. An alternative formulation related to cyclotomic polynomials is given in \cite{BREUER2007293}, Proposition 3.6.
\end{proof}
%

%
\begin{lemma}
Let k and n be positive integers. Then, the binomial coefficient $\binom{2^{n}}{k}$ is even for all $0<k<2^{n}$.
\label{lem:binimial_2^n}
\end{lemma}
Several proofs for this property are known, the following is perhaps the simplest one, adapted from \cite{schepler2019_3087725} with minor changes.
\begin{proof}
According to the binomial theorem, 
\[
(x+1)^{2^{n}} = \sum_{k=0}^{2^{n}} \binom{2^{n}}{k} x^{k}
\,\,.
\]
Therefore the lemma will be established if we prove that for all positive integer $n$ there is some integer $a_k$ for each $0<k<2^{n}$ such that
\begin{align}
(x+1)^{2^{n}} = x^{2^{n}} + 1 + \sum_{k=1}^{2^{n}-1} 2a_k\, x^{k}
\,\,.
\label{induction_identity}
\end{align}
We will prove this by induction on $n$. First we prove the base case where $n=1$, so that 
\[
(x+1)^{2^{n}} = (x+1)^{2} = x^2 + 1 + 2x
\,\,.
\]
which has the form of \eqref{induction_identity} with $a_1 = 1$.

The induction hypothesis is that this property holds for $n = m$, for a given $m \ge  1$, so that 
\[
(x+1)^{2^{m}} =  x^{2^{m}} + 1 + \sum_{k=1}^{2^{m}-1}2 a_k\, x^{k}
\,\,.
\]
We now prove the property holds for $m+1$:
\begin{align*}
(x+1)^{2^{m+1}} 
&= \left((x+1)^{2^{m}}\right)^2
\end{align*}
By the induction hypothesis this is 
\begin{align*}
\left((x+1)^{2^{m}}\right)^2 &=
\left(x^{2^{m}} + 1 + \sum_{k=1}^{2^{m}-1} 2 a_k\, x^{k}\right)^2
\\
&=
\left\{
\left(x^{2^{m}} + 1\right)^2 
\right\}
+
\left\{
2 \left(x^{2^{m}} + 1\right)\left(2\sum_{k=1}^{2^{m}-1} a_k\, x^{k}\right)
\right\}\\
& \qquad +
\left\{
\left(2\sum_{k=1}^{2^{m}-1} a_k\, x^{k}\right)^2
\right\}
\\
&=
x^{2^{m+1}} + 1 + 2\, x^{2^{m}} 
\\
& \quad 
+ 4 \left\{\left(x^{2^{m}} +1 \right) \sum_{k=1}^{2^{m}-1} a_k\, x^{k} +\left(\sum_{k=1}^{2^{m}-1} a_k\, x^{k}\right)^2\right\}
\end{align*}
The first two terms to the right of the final equality are the first two terms of \eqref{induction_identity} with $n =m+1$. The binomial theorem applied to $(x+1)^{2^{m+1}}$ produces these same two terms, and then a polynomial in $x$ with powers running from $1$ to $2^{m+1}-1$, inclusive. Because all terms in this particular remainder are multiplied either by two or by four, they are all explicitly even, and the resulting polynomial will be sums of even numbers, which must therefore also be even. So it is possible to find numbers $b_k$ such that
\[
(x+1)^{2^{m+1}} =  x^{2^{m+1}} + 1 + \sum_{k=1}^{2^{m+1}-1}2 b_k\, x^{k}
\,\,.
\]
In particular,
\[
 b_{k}=\begin{cases}
        2a_{1}, & \text{for $k=1$;}\\
        2a_{k}+\sum_{m=1}^{k-1}a_{k-m}a_{m}, & \text{for $k=2,\,\ldots,\,2^{m}-1$;}\\
        1+2\sum_{m=1}^{2^{m}-1}a_{2^m-m}a_{m},  & \text{for $k=2^m$;}\\
        2a_{k-2^m}+\sum_{m=k-2^m+1}^{2^m-1}a_{m}a_{k-m}, & \text{for $k=2^m + 1,\,\ldots,\,2^{m+1}-2$; and}\\
        2a_{2^m-1}, & \text{for $k=2^{m+1}-1$.}
       \end{cases}
\]

\end{proof}
%

%

\section{Fundamental periods and a propagator for the Rule 60 on a ring of a Mersenne number circumference}
\subsection{Fundamental periods}
\begin{definition}
 A domain is the maximal contiguous set of sites having the same occupation state bounded by a right domain boundary and a left domain boundary.
\end{definition}
\begin{definition}
    A site $x_{i,j}$ is said to be a right domain boundary (RDB) if the site $x_{i+1,j}$ has a different occupation than $x_{i,j}$. Similarly, a site $x_{i,j}$ is said to be a left domain boundary (LDB) if the site $x_{i-1,j}$ has a different occupation. If $x_{i,j}$ is both a RDB and a LDB, then it is a single-site domain. If $x_{i,j}$ is neither a RDB or a LDB, then it is an interior site of a domain.
\end{definition}
\begin{definition}
    A bit flip on $x_{i,j}$ is an operation which changes the occupation of that site, either from $1$ to $0$ or from $0$ to $1$, and leaves all other sites in $\vec{x}_j$ unchanged.
\end{definition}
\begin{definition}
    A maximal state is a state of length $L$ with the maximum number of domains.
\end{definition}
\begin{lemma}
Every ring of odd circumference $L$ has an even number of domains for all occupation states.
\label{lem:even_domains} 
\end{lemma}
\begin{proof}
    Since $L$ is odd, it is not possible to have only single-site domains on a ring. For all states of odd $L$, at least one domain must contain at least two sites. For the maximal state, only one domain will contain two sites, all others will be a single-state domain. Thus with periodic boundary conditions, when $L$ is odd, the maximum number of domains is $L-1$, which is even.
    
    Any other state can be made from a maximal state by a series of bit flips on successive sites, up to $L$ sites. On each site there are four possible cases:\\

    Case 1: The site $x_{i,j}$ is a single-site domain. By definition, this means the occupation states of $x_{i-1,j}$ and $x_{i+1,j}$ are the same. Therefore when the $x_{i,j}$ site is bit flipped, its single-site domain will merge with the two domains to the right and to the left, reducing the total number of domains by two (three original domains become one). A maximal state has an even number of domains, and bit flipping any of the single site domains results in a net loss of two domains, so the total number of domains remains even.\\
    \\    
    
    Case 2: The site $x_{i,j}$ is an interior site of a domain. This is essentially the reverse of Case 1; when this site is bit flipped, it will then become a single site domain, with two new domains on either side. The original single domain is therefore divided into three new domains by the bit flip, for a net gain of two domains. Since the number of domains in a maximal state is even, if a site interior to a domain is bit flipped, two domains are added and the total remains even.\\
    \\

    Case 3: The site $x_{i,j}$ is a LDB, but not a RDB. By definition, since $x_{i,j}$ is a LDB but not a RDB it cannot be a single-site domain. Bit flipping this site will therefore shift the domain boundary to the right by one site, merging the $x_{i,j}$ site with the domain that previously had the RDB at $x_{i-1,j}$ and creating a new LDB at $x_{i+1,j}$. The total number of domains remains the same. Since the maximal state has an even number of domains, bit flipping a site which is a LDB but not a single-site domain results in the same even number of domains.\\
    \\

    Case 4: The site $x_{i,j}$ is a RDB, but not a LDB. Again by definition, this site cannot be a single-site domain. Similarly to Case 3, bit flipping this site will result in the domain boundary shifting left one site, making $x_{i,j}$ the RDB previously at $x_{i+1,j}$. The total number of domains again remains the same. Therefore all possible cases of bit flips from the maximal state result in a state with an even number of domains.
\end{proof}
\noindent
{\bf Remark}: The above proof can be extended to even $L$ to formulate an alternate proof of Lemma \ref{lem:odd_and_even_to_even}.
\begin{definition}
    A predecessor of $\vec{x}_{j+1}$ is any state $\vec{x}_{j}$ that evolves to $\vec{x}_{j+1}$ under the time evolution specified in \eqref{rule_60}
\end{definition}
\begin{lemma}
    A state on a ring of circumference $L$ has a predecessor if, and only if, it is even. 
    \label{lem:predecessor_exist}
\end{lemma}
\begin{proof}
    We have already shown in Lemma \ref{lem:odd_and_even_to_even} that all time evolution under Rule 60 results in even states, so no $_L(\text{odd})$ state has a predecessor. If a predecessor does exist, it must be constructible by the following inverse of Rule 60. Here we construct an $\vec{x}_j$ configuration from the knowledge of the $\vec{x}_{j+1}$ configuration. This inverse rule is 
\begin{align}
\begin{array}{lcc}
x_{0,j} = 1 &\\
x_{i,j} = \text{Xor}(x_{i-1,j},\,x_{i-1,j+1}) & \text{ for } & i=1,\,2,\,\ldots,\,L-1.
\end{array}
\label{predecessor_rule_60}
\end{align}
Because of the periodic boundary condition, there is also the self-consistency condition
\begin{equation}
    x_{0,j} = \text{Xor}(x_{L-1,j},\,x_{L-1,j+1}).
    \label{predecessor_rule_condition}
\end{equation}
Note that there is a freedom in the initial choice of $x_{0,j}$, we have chosen $1$ for this proof, but one can equally choose $0$ and the rest of the construction goes through as above. \\
\\
We will now prove that all even states have a predecessor by induction on $L$. \\
    \\
    For the base case, we take $L=2$ at some time $j+1$. There are four possible configurations:
    \begin{align*}
        &[00]\\
        &[01]\\
        &[10]\\
        &[11]
    \end{align*}
Applying the inverse rule \eqref{predecessor_rule_60}, we find:
\begin{center}
    $[00]$ has predecessor $[11]$\\
    $[11]$ has predecessor $[10]$\\
   \end{center}
The configurations $[01]$ and $[10]$ do not have predecessors as application of the inverse rule \eqref{predecessor_rule_60} cannot satisfy the self-consistency condition in both cases. Thus, we have proven by demonstration that all even states for $L=2$ have a predecessor state, and no odd states have a predecessor.

The induction hypothesis is that this holds for some $L=N$, so that given an $\vec{x}_{j+1}$ of length $N >2$
\begin{align*}
    x_{i,j} &= \text{Xor}(x_{i-1,j},\,x_{i-1,j+1}) & \text{ for }  i=1,\,2,\,\ldots,\,N-1\\
    & \text{Xor}(x_{N-1,j},\,x_{N-1,j+1})  = 1 &\text{for} _N(\text{even})\\
     & \text{Xor}(x_{N-1,j},\,x_{N-1,j+1})  = 0 &\text{for} _N(\text{odd})
\end{align*}
We now prove the property holds for $N+1$. There are two cases:\\
Case 1: The ring of circumference $N$ had an even occupation, so by the induction hypothesis
\begin{equation*}
    \text{Xor}(x_{N-1,j},\,x_{N-1,j+1}) = 1.
\end{equation*}
Extending the ring to $N+1$ means adding another site, $x_{N,j+1}$, which can be either occupied or unoccupied. If the site is occupied then the total occupation of $N+1$ is odd. By \eqref{predecessor_rule_60}
\begin{equation*}
  x_{N,j} = \text{Xor}(x_{N-1,j},\,x_{N-1,j+1}) = 1
\end{equation*}
but by the self-consistency condition \eqref{predecessor_rule_condition}
\begin{equation*}
    x_{0,j} = \text{Xor}(x_{N,j},\,x_{N,j+1}) = 0 
\end{equation*}
which is a contradiction, and so there can be no predecessor state. 

If $x_{N,j+1}$ is unoccupied, then again 
\begin{equation*}
  x_{N,j} = \text{Xor}(x_{N-1,j},\,x_{N-1,j+1}) = 1
\end{equation*}
but now 
\begin{equation*}
    x_{0,j} = \text{Xor}(x_{N,j},\,x_{N,j+1}) = \text{Xor}(1,0)= 1
\end{equation*}
and the predecessor state has been successfully constructed using rule \eqref{predecessor_rule_60}.

Case 2: The ring of circumference $N$ had an odd occupation, so by the induction hypothesis
\begin{equation*}
    \text{Xor}(x_{N-1,j},\,x_{N-1,j+1}) = 0.
\end{equation*}
If the additional site is occupied, so $x_{N,j+1}=1$ then the occupation of $N+1$ is now even, and 
\begin{equation*}
    x_{0,j} = \text{Xor}(x_{N,j},\,x_{N,j+1}) = \text{Xor}(1,0)= 1
\end{equation*}
and the predecessor state has been successfully constructed. 

If $x_{N,j+1}=0$, then the $N+1$ state remains odd and 
\begin{equation*}
    x_{0,j} = \text{Xor}(x_{N,j},\,x_{N,j+1}) = \text{Xor}(0,0)= 0
\end{equation*}
in contradiction to \eqref{predecessor_rule_condition} and a predecessor state can not be constructed. Therefore, the ring of circumference $N+1$ has a predecessor state only if the occupation is even.
\end{proof}
\begin{lemma}
    Every $_L(\text{even})$ state on a ring has two possible predecessors, which are related by bit flipping every site.
    \label{lem:two_precurosors}
\end{lemma}
\begin{proof}
    In Lemma \ref{lem:predecessor_exist} we demonstrated that every $_L(\text{even})$ state has a predecessor, but noted that there is a freedom in choosing the occupation of the initial site of the predecessor state. Since the Xor operation only considers relationships between pairs of sites, $x_{i,j}$ and $x_{i-1,j}$, bit flipping every site preserves the outcome of the Xor operation and generates a second predecessor. Since this is the only freedom in constructing the predecessor states, these two states are the only predecessor states.
\end{proof}
\begin{lemma}
On a ring of odd circumference $L$, every $_{L}(\text{even})$ state has one 
$_{L}(\text{odd})$ precursor.  
\label{lem:odd_precursor}
\end{lemma}
\begin{proof}
Construct an $\vec{x}_{j}$ state from a given $\vec{x}_{j+1}$ by applying \eqref{predecessor_rule_60}. If the occupation is odd, then the proof is complete. If the occupation is even, then by Lemma \ref{lem:two_precurosors} the second state can be found by bit flipping every site of the even state. Since $L$ is odd, this flipped state will be odd.
\label{odd_precursor}
\end{proof}
\begin{definition}
    A Mersenne number, $\mathcal{N}$ is any number where $\mathcal{N}= 2^n-1$ for some integer $n >1$.
\end{definition}
\begin{theorem}
For a ring whose circumference is a Mersenne number, $\mathcal{N}$, every $_{\mathcal{N}}(\text{even})$-type initial state initiates an $L$-periodic trajectory. 
\label{th:2^n-1}
\end{theorem}
\noindent
{\bf Remark}: The proof that follows is not new. The most elegant one can be found in 
\cite{ehrlich1990_302} Lemma 1, Corollary 1. Our proof differs in several details from the one in \cite{ehrlich1990_302}. In particular, it does not rely on the properties of the propagators, and as such is technically close to the one given in \cite{sfairopoulos2023_174107} for $L=2^n$. 
\noindent
\begin{proof}

Let $\hat{I}$ be the identity matrix of appropriate dimension. Let $\hat{S}_{L}$ be a left shift operator for a ring of circumference $L$ such that
\begin{align*}
\begin{array}{lcc}
(\hat{S}_{L} \vec{x})_{i} = (\vec{x})_{i-1} & \text{ for } & i=1,\,2,\,\ldots,\,L-1
\\
(\hat{S}_{L} \vec{x})_{0} = (\vec{x})_{L-1} &\text{for} &i=0. \qquad \qquad \qquad
\end{array}
\end{align*}
The shift operator has the property:
\begin{align}
\left(\hat{S}_{L}\right)^{L} = \left(\hat{S}_{L}\right)^{0} = \hat{I}
\,\,.
\label{S_L_to_L}
\end{align}

The state of the system at an instance of time $j$ given by the vector $\vec{x}_{j}$. Observe that
\begin{align}
\vec{x}_{j} \stackrel{\mod 2}{=} (\hat{I}+\hat{S}_{L})^{j} \vec{x}_{0}
\,\,,
\label{time_shift}
\end{align}

Now, consider $L=\mathcal{N}=2^{n}-1$. Apply Rule 60 to some initial state $\vec{x}_{0}$ $L+1=2^n$ times. By \eqref{time_shift} the result can be expressed as 

\begin{align*}
\vec{x}_{L+1} \stackrel{\mod 2}{=} (\hat{I}+\hat{S}_{L})^{2^n} \vec{x}_{0}
\,\,.
\end{align*}
By the binomial theorem, the quantity $(\hat{I}+\hat{S}_{L})^{2^n}$ can be expanded as

\begin{align*}
\vec{x}_{L+1} \stackrel{\mod 2}{=} \left( \sum_{k=0}^{2^n}\binom{2^n}{k}\hat{I}^{2^n-k} \hat{S}_{L}^k\right) \vec{x}_{0}
\,\,.
\end{align*}
Because  $\hat{I}$ is the identity, $\hat{I}^{2^n-k} =\hat{I}$ for all $k$ and can be factored out of the sum: 
\begin{align*}
\vec{x}_{L+1} \stackrel{\mod 2}{=} \left( \hat{I}\sum_{k=0}^{2^n} \binom{2^n}{k}\hat{S}_{L}^k\right) \vec{x}_{0} = \left(\sum_{k=0}^{2^n} \binom{2^n}{k}\hat{S}_{L}^k\right) \vec{x}_{0}
\,\,.
\end{align*}
When $k=0$ no shift has been applied, and by \eqref{S_L_to_L} $\left(\hat{S}_{L}\right)^{0} = \hat{I}$, so this first term can be explicitly evaluated and removed from the sum:
\begin{align*}
\vec{x}_{L+1} \stackrel{\mod 2}{=}  \left( \hat{I}+ \sum_{k=1}^{2^n} \binom{2^n}{k}\hat{S}_{L}^k\right) \vec{x}_{0}
\,\,.
\end{align*}

By Lemma \ref{lem:binimial_2^n} $\binom{2^n}{k}$ is even for all $0< k < 2^n$, and so all coefficients except $\binom{2^n}{2^n}$ are zero mod 2, leaving 
\begin{align*}
\vec{x}_{L+1} \stackrel{\mod 2}{=} \left( \hat{I}+ \hat{S}_{L}^{2^n}\right) \vec{x}_{0} =\left( \hat{I}+ \hat{S}_{L}^{L+1}\right) \vec{x}_{0}
\,\,.
\end{align*}
We now note 
\begin{equation*}
\hat{S}_{L}^{L+1} = \hat{S}_{L}^{L} \hat{S}_{L} \,\, ,
\end{equation*}
but from \eqref{S_L_to_L}, $\left(\hat{S}_{L}\right)^{L}=\hat{I}$, so 
\begin{align*}
\vec{x}_{L+1} \stackrel{\mod 2}{=} \left( \hat{I}+ \hat{S}_{L}\right) \vec{x}_{0}
\,\,,
\end{align*}
which by \eqref{time_shift} is the statement 
\begin{align*}
\vec{x}_{L+1} = \vec{x}_{1}
\,\,.
\end{align*}
By Definition \ref{def:period}, this is an $L$-periodic trajectory. Since by Corollary \ref{cor:no_periods_for_odd}, no $_L(odd)$-type initial state can initiate a periodic trajectory, this holds for all $_{\mathcal{N}}(\text{even})$-type initial states.
\end{proof}
\begin{corollary}
    For a ring whose circumference is a Mersenne number, $L$ is the maximal fundamental period.
    \label{cor:Mersenne_max_period}
\end{corollary}
\begin{proof}
    By Theorem \ref{th:2^n-1}, every $_L(\text{even})$-type initial state initiates an $L$ periodic motion when $L$ is a Mersenne number. By Lemma \ref{lem:odd_and_even_to_even}, every odd occupation state becomes even in one forward time step. In particular, this means that any $_L1$ initial state will become $_L(\text{even})$ after one step, and initiate an $L$ periodic motion. By Lemma \ref{lem:maximal_M_and_shorter_periods}, this means $L$ is a maximal fundamental period, and all other fundamental periods divide $L$.
\end{proof}
\begin{corollary}
For a ring whose circumference is a Mersenne prime, every 
$_{L}(\text{even})$-type initial state other than the state with zero occupied sites initiates a $L$-fundamental-periodic trajectory. 
\end{corollary}
\begin{proof}
By Theorem \ref{th:2^n-1}, every $_L(\text{even})$-type initial state initiates an $L$ periodic motion when $L$ is a Mersenne number, and by Corollary \ref{cor:Mersenne_max_period} $L$ is a maximal fundamental period, and all other fundamental periods divide $L$. In this case where $L$ is a Mersenne prime, there are no factors of $L$ besides $L$ and one, so $L$ is the fundamental period. By Lemma \ref{lem:all_0s}, the excluded all-zero state initiates the period one trajectory.

\end{proof}
\subsection{Propagator}
To aid in the aim of constructing all ground states of the Newman-Moore model, let us first introduce a propagator, i.e.\ a solution of \eqref{rule_60} that has only one site occupied at $j=0$. Namely let us introduce 
\begin{align}
\mathcal{P}_{(i_{},j_{}) \leftarrow (i_{0},0)}
\label{propagator}
\end{align}
that obeys 
\begin{align}
\begin{array}{lcc}
\quad\mathcal{P}_{(i,0)\leftarrow (i_{0},0)} = \delta_{i_{},i_{0}} 
\\
\mathcal{P}_{(i,j+1)\leftarrow (i_{0},0)} = 
\text{Xor}(\mathcal{P}_{(i-1,j)\leftarrow (i_{0},0)},\,\mathcal{P}_{(i,j)\leftarrow (i_{0},0)}) \\\qquad \qquad \qquad \text{ for } \hspace{0.25em} i=1,\,2,\,\ldots,\,L-1
\\
\mathcal{P}_{(0,j+1)\leftarrow (i_{0},0)} =\text{Xor}(\mathcal{P}_{(L-1,j)\leftarrow (i_{0},0)},\,\mathcal{P}_{(0,j)\leftarrow (i_{0},0)})& &
\\
\quad j=0,\,1,\,\ldots
\\
\quad i_{0}=0,\,1, \ldots ,\,L-1
\end{array}
\label{rule_60__for_propagator}
\end{align}
\begin{theorem} \label{prop_thrm}
For $L=2^n-1$ ($n=2,\,3,\,\ldots$),
\begin{align}
   & \mathcal{P}_{(i,j)\leftarrow (0,0)} =
    \begin{cases}
        \binom{j}{i}_L \!\!\!\! \mod{2}, \quad 0 \leq j \leq L-1\\
        (1-\delta_{i,0}), \quad j = L\\
        \mathcal{P}_{(i,(j-1 \!\!\!\! \mod L)+1)\leftarrow (0,0)}, \quad j > L
    \end{cases}
\label{explicit_propagator}
\end{align}

\label{th:propagator}
\end{theorem}
where $ \binom{j}{i}_L$ is the q-binomial coefficient defined as 
\begin{equation}
    \binom{n}{m}_q = \prod_{k=0}^{m-1} \frac{1-q^{n-k}}{1-q^{k+1}}
\end{equation}
\begin{proof}
For the initial state, note that
\begin{align}
\binom{0}{i}_L &= 
\begin{cases}
1 \quad i=0\\
0 \quad \text{otherwise}
\end{cases}
\end{align}
so the $j=0$ state has a single occupied site at $i=0$ as desired.

As stated in \eqref{rule_60}, , the time evolution is given by
\begin{align}
x_{i,j+1} = \text{Xor}(x_{i-1,j},\,x_{i,j}), & \quad 0 < j \leq L-1.
\end{align}
This is equivalent to addition mod 2, so 
\begin{equation}
x_{i,j+1} = \left[x_{i-1j} + x_{ij} \right]\mod{2}.
\label{rule_60_addition}
\end{equation}
We make use of the known recursion relation for the q-binomial coefficients
\begin{equation}
\binom{j+1}{i}_L = \binom{j}{i-1}_L + L^i \binom{j}{i}_L.
\label{qbinomial_recursion}
\end{equation}
Note that for $L = 2^n -1$
\begin{equation}
    L^i \mod{2} = 1 \quad \forall i
\end{equation}
since all terms in the expansion of $(2^n -1)^i$ will be powers of two and their multiples, except the $\binom{i}{i} \left(2^{n}\right)^{0}(-1)^{i}$ term. Thus
\begin{equation}
\binom{j+1}{i}_L \mod{2} = \left[\binom{j}{i-1}_L + \binom{j}{i}_L\right] \mod{2},
\label{qbinomial_sum}
\end{equation}
which is equivalent to \eqref{rule_60_addition}. 
\end{proof}

For other propagators, the translational invariance dictates 
\begin{theorem}
\begin{align}
\begin{split}
&
\mathcal{P}_{(i,j)\leftarrow (i_{0},0)} = 
\left\{
\begin{array}{lcc}
\mathcal{P}_{(i+i_{0},j)\leftarrow (i_{0},0)} & \text{\rm for } & i \leq 2(L-1)-i_{0}
\\
\mathcal{P}_{(i+i_{0}-(L-1),j)\leftarrow (i_{0},0)} & \text{\rm for } & i> 2(L-1)-i_{0}
\end{array}
\right.
\\
&
i_{0} = 0,\,1\,\ldots,\,L-1
\\
&
i = 0,\,1\,\ldots,\,L-1
\\
&
j = 0,\,1\,\ldots,\,L-1
\end{split}
\label{all_explicit_propagators}
\end{align}
\end{theorem}
\begin{proof}
This follows directly from the definition of the propagator. 
\end{proof}
The propagator \eqref{all_explicit_propagators} can be used to propagate the Rule 60 automaton with $L=2^{n}-1$:
\begin{align*}
&
x_{i,j}=\sum_{i_{0}=0}^{L-1} \mathcal{P}_{(i,j)\leftarrow (i_{0},0)} x_{i_{0},j}
\\
&
i = 0,\,1\,\ldots,\,L-1
\\
&
j = 0,\,1\,\ldots,\,L-1
\end{align*}
Note that Equations \eqref{explicit_propagator} and \eqref{all_explicit_propagators} are not directly applicable for $L\neq 2^{n}-1$. 

\section{Application to the Newman-Moore model}
%
%
\subsection{Connection between the Rule 60 automaton and the Newman-Moore model}
To connect the Rule 60 automaton and the two-dimensional spin lattices, one considers the so-called triangular plaquette model or Newman-Moore model \cite{newman1999_5068,eczoo}:
\begin{align}
H = -J\sum_{i=0}^{L-1} \sum_{j=0}^{M-1} 
\sigma_{i,j} \sigma_{i+1,j} \sigma_{i+1,j+1}
\label{H}
\end{align}
with periodic boundary conditions. Notice that according to Rule 60, it is always the case that 
\begin{align*}
(1-2x_{i,j})(1-2x_{i+1,j})(1-2x_{i+1,j+1}) = 1
\,\,.
\end{align*}
Hence, every $M$-periodic trajectory of \eqref{rule_60} will produce a ground state of the Hamiltonian \eqref{H}, whose energy is 
\begin{align*}
E_{\text{ground}} = E_{\uparrow\ldots\uparrow} = -JL M
\,\,,
\end{align*}
if one uses the following association: \[\sigma_{i,j}=(1-2x_{i,j})\,\,.\] Notice that ``all spins up'' is one of the ground states.

%
\subsection{Newman-Moore results inferred from the properties of the Rule 60 automaton}
We are now ready to address the principal application of our cellular automata result: the Newman-Moore model. In what follows, we will be able to generate all  ground state configurations of the Newman-Moore model \eqref{H} with 
\begin{align*}
&
M=L=2^n-1
\\
&
n=2,\,3,\,\ldots
\,\,,
\end{align*}
i.e. for a square Newman-Moore lattice with the sides equal to one of the Mersenne numbers (not necessarily prime). 

We assert the following.
\begin{theorem}
The ground state of the Newman-Moore model \eqref{H} with $M=L=2^n-1$ is $2^{L-1}$ degenerate. 
\label{th:N-M_degeneracy}
\end{theorem}
\begin{proof}
The proof follows directly from Theorem \ref{th:2^n-1}.
\end{proof}

We will now provide explicit expressions for the ground state configurations of the  Newman-Moore model \eqref{H} on an $L\times L$ square lattice with $L$ being equal to one of the Mersenne numbers $2^{n}-1$.  

The ground state configurations will be labeled by the state of the $0$-th row:
\[
[\sigma_{0,0},\,\sigma_{1,0},\,\ldots,\,\sigma_{L-1,0}]
\,\,,
\]
where number of spin-down sites, $\sigma=-1$ is even: there are $2^{L-1}$ such configurations, consistent with the $2^{L-1}$-fold degeneracy.
\begin{theorem}
The ground states of the  Newman-Moore model \eqref{H} on an $L\times L $ square lattice with $L=(2^{n}-1)$ ($n=2,\,3,\,\ldots$) read 
\begin{align}
\begin{split}
&
\sigma_{i,j}=\sum_{i_{0}=0}^{L-1} 
\left(1-2\mathcal{P}_{(i,j)\leftarrow (i_{0},0)}
\frac{1-\sigma_{i_{0},j}}{2}\right)
\\
&
i = 0,\,1\,\ldots,\,2^n-2
\\
&
j = 0,\,1\,\ldots,\,2^n-2
\end{split}
\label{ground_states}
\end{align}
where the propagator 
$\mathcal{P}_{(i,j)\leftarrow (i_{0},0)}$ is given by 
\eqref{all_explicit_propagators} and 
$[\sigma_{0,0},\,\sigma_{1,0},\,\ldots,\,\sigma_{2^n-2,0}]$ is any of the $2^{L-1}$ spin configurations characterised by the number of the $\sigma=-1$ sites being even.
\label{th:N-M_states}
\end{theorem}
%

\section{Summary of results}
Theorem \ref{th:2^n-1} plays the central role in our paper. It allows to identify all the initial conditions that lead to periodic trajectories in a Rule 60 cellular automaton on a ring whose circumference $L$ is given by one of the Mersenne numbers, $L=2^{n}-1$ . Furthermore, all periodic trajectories are shown to share the period $L$. The Theorem \ref{th:propagator} constructs the propagators for this model. 

The above results allow us to list all the ground states of the Newman-Moore model. These results are summarized in Theorems \ref{th:N-M_degeneracy} and \ref{th:N-M_states}

\section*{Acknowledgements}
The authors wish to thank K. Sfairopoulos for bringing this problem to their attention, for providing invaluable guidance on references, and for many useful discussions.

A significant portion of this work was produced during the thematic trimester on ``Quantum Many-Body Systems Out-of-Equilibrium'', at the  Institut Henri Poincaré (Paris): MO  is immeasurably grateful to the organizers of the trimester, Rosario Fazio, Thierry Giamarchi, Anna Minguzzi, and Patrizia Vignolo,  for an opportunity to be a part of it. MO and JC-P wish to express their gratitude to the International High IQ Society for the networking opportunities that it offers. 

\paragraph*{Funding information} 
MO and JR were supported by the NSF Grant No.~PHY-2309271. MO would like to thank the Institut Henri Poincar\'{e} (UAR 839 CNRS-Sorbonne Université) and the LabEx CARMIN (ANR-10-LABX-59-01) for their support.

\bibliography{JoCA/TEX-Style-Guide/Rule60NM}

\end{document}